\documentclass[11pt,conference,compsoc,onecolumn,romanappendices]{IEEEtran}
\usepackage[cmex10]{amsmath}
\usepackage{amssymb,enumerate,bbm}
\usepackage{float}
\usepackage{cite}
\usepackage{setspace}
\usepackage{fancyhdr}
\fancypagestyle{plain}{%
  \fancyhf{}
  \fancyfoot[C]{\iffloatpage{}{\thepage}}
  }
\usepackage{amsthm}
\theoremstyle{definition}

\newcommand{\nin}{\not\in}
\newcommand{\nobracket}{}
\newcommand{\nocomma}{}
\newcommand{\noplus}{}
\newcommand{\tmem}[1]{{\em #1\/}}
\newcommand{\tmname}[1]{\textsc{#1}}
\newcommand{\tmop}[1]{\ensuremath{\operatorname{#1}}}
\newcommand{\tmsamp}[1]{\textsf{#1}}
\newcommand{\tmstrong}[1]{\textbf{\textit{#1}}}

\newenvironment{enumeratealpha}{\begin{enumerate}[a{\textup{)}}] }{\end{enumerate}}
\newenvironment{enumeratenumeric}{\begin{enumerate}[1.] }{\end{enumerate}}
\newenvironment{enumerateromancap}{\begin{enumerate}[I.] }{\end{enumerate}}
\newtheorem{corollary}{\textbf{Corollary}}
\newtheorem{definition}{\textbf{\textit{Definition}}}
\newtheorem{lemma}{\noindent\textbf{Lemma}}
\newtheorem{theorem}{\textbf{Theorem}}

\floatstyle{boxed}
\newfloat{program}{htp}{lop}
\floatname{program}{Program}
\newcommand{\code}[2]{  
  \begin{program}
  #1
  \vspace*{1mm}
  \hrule
  #2
  \end{program}
}
\IEEEoverridecommandlockouts

\title{\LARGE Asynchronous Byzantine Agreement with Optimal Resilience and Linear Complexity}

\author{
  \IEEEauthorblockN{Cheng Wang}
  \IEEEauthorblockA{
    EPFL \\
    cheng.wang@epfl.ch
  }
}

\begin{document}

\pagenumbering{gobble}

\maketitle

\begin{abstract}
\emph{
Given a system with $n > 3t + 1$ processes, where $t$ is the tolerated number of faulty ones, we present a fast asynchronous Byzantine agreement protocol that can reach agreement in $O(t)$ expected running time.
This improves the $O(n^2)$ expected running time of Abraham, Dolev, and Halpern {\cite{abraham2008almost}}.
Furthermore, if $n = (3 + \varepsilon) t$ for any $\varepsilon > 0$, our protocol can reach agreement in $O (1 / \varepsilon)$ expected running time.
This improves the result of Feldman and Micali {\cite{feldman1988optimal}} (with constant expected running time when $n > 4 t$).
}
\end{abstract}

\newpage
\pagenumbering{arabic}\setcounter{page}{1}

\section{Introduction}

The {\tmem{Byzantine Agreement}} (BA) problem,
first introduced by Pease, Shostak, and Lamport {\cite{pease1980reaching,lamport1982byzantine}}, is a fundamental problem in distributed computing.
Given $n$ processes, $ t $ of which being faulty, the problem 
consists for all correct processes to agree on one of the input values.
The faulty processes might deviate 
from the algorithm assigned to them arbitrarily, e.g., 
to prevent correct processes from agreeing on one of their input values.

A lot of work has been devoted to the problem in the last three decades. Despite 
the effort,
the {\tmem{Asynchronous Byzantine Agreement}} (ABA)
problem, where the communication between processes can take an arbitrary  
amount of time, is still not very well understood.
Certain results are however known. For example, it is known that the problem is 
impossible to solve if $n \leqslant 3 t$
{\cite{pease1980reaching,karlin1986probabilistic}}. 
Any ABA protocol  assuming $n > 3 t$ is called {\tmem{optimally resilient}}. 
According to the seminal result of
{\cite{fischer1985impossibility}}, any deterministic ABA protocol
must have some non-terminating execution. 

Faced with the impossibility result {\cite{fischer1985impossibility}}, a natural direction of research is to design efficient {\tmem{randomized}} Byzantine agreement protocol.
This direction was started with the work of Ben-or {\cite{ben1983another}}, Rabin
{\cite{rabin1983randomized}}, and Bracha {\cite{bracha1984asynchronous}}.
Remarkably, Canetti and Rabin {\cite{canetti1993fast}} proposed an ABA protocol with constant expected running time and overwhelming probability to terminate.
With a randomized ABA protocol the best that can be achieved is to have every execution terminate with 
probability one. Such protocols are said to be {\tmem{almost-surely
terminating}} {\cite{abraham2008almost}}. 

Several almost-surely terminating ABA protocols were 
proposed.  In 1983, Ben-Or
{\cite{ben1983another}} proposed an almost-surely terminating ABA protocol
for $n > 5 t$, which runs in exponential expected time.  One year later, 
Bracha {\cite{bracha1984asynchronous}} 
presented an almost-surely terminating ABA, which also runs in exponential expected time, but with 
optimal resilience, i.e., for $n > 3 t$. In 1988, Feldman and Micali {\cite{feldman1988optimal}}  
presented an almost-surely
terminating ABA protocol with constant expected time, assuming however $n > 4 t$. 
Twenty years later, Abraham, Dolev, and Halpern {\cite{abraham2008almost}} presented an 
almost-surely terminating optimally resilient ABA protocol with  
polynomial efficiency (the expected running time is $O (n^2)$). 
In some sense, state-of-the-art results for almost-surely
terminating ABA  are {\cite{feldman1988optimal}} and
{\cite{abraham2008almost}}:  optimally resilience with  
polynomial efficiency on the one hand, 
or constant expected time, assuming however $n > 4 t$, on the other hand. 

We present in this paper a new almost-surely terminating ABA protocol that achieves a significant progress 
with respect to the state-of-the-art.
For $n > 3 t$, our protocol completes in $O (t)$ expected running time. If $n > (3 +
\varepsilon) t$ where $\varepsilon$ is an arbitrary positive constant, our
protocol has $O (1 / \varepsilon)$ expected running time. 
Table \ref{table-results} aggregates these results in the context of related work. 

\begin{table}[b]
  \centering
  \normalsize
  \begin{tabular}{|c|c|c|}
    \hline
    Reference &  Resilience  & Expected Running Time\\
    \hline
    Ben-Or (1983) {\cite{ben1983another}} & $n > 5 t$ & $O (2^n)$\\
    \hline
    Bracha (1984) {\cite{bracha1984asynchronous}} & $n > 3 t$ & $O (2^n)$\\
    \hline
    Feldman, Micali (1988) {\cite{feldman1988optimal}} & $n > 4 t$ & $O (1)$\\
    \hline
    Abraham, Dolev, Halpern (2008) {\cite{abraham2008almost}} & $n > 3 t$ & $O
    (n^2)$\\
    \hline
    This paper & $n > 3 t$ & $O (t)$\\
    \hline
    This paper & $n > (3 + \varepsilon) t, (\varepsilon > 0)$ & $O (1 / \varepsilon)$\\
    \hline
  \end{tabular}
  \medskip

  \caption{\label{table-results}Results for almost-surely terminating ABA problem}
\end{table}

Most ABA  protocols follow the idea of Ben-or {\cite{ben1983another}}, Rabin
{\cite{rabin1983randomized}}, and Bracha {\cite{bracha1984asynchronous}}, namely a reduction of the 
ABA problem to the implementation of a {\tmem{common coin}}
(namely, a source of common randomness with certain properties). Specially,
the reduction of Bracha {\cite{bracha1984asynchronous}} is optimally resilient
and runs in constant expected time. Thus, designing efficient ABA protocols
could be solved by designing efficient common coins. The protocol of Feldman
and Micali {\cite{feldman1988optimal}} includes a method to implement a common
coin by making use of a {\tmem{verifiable secret sharing}} (VSS) scheme. (For a complete
description of the reduction from VSS to ABA, see 
{\cite{canetti1996studies}}.) Canetti and Rabin {\cite{canetti1993fast}}
have an implementation of {\tmem{asynchronous verifiable secret sharing}}
(AVSS) with constant expected running time but overwhelming probability
to terminate (the resulting ABA protocol is thus not almost-surely
terminating). Recently, King and Saia {\cite{king2014faster}} introduced a novel technique
for implementing common coin via a spectral method.

This paper follows the reduction from (some form of) AVSS to ABA.
We first recall the standard AVSS scheme {\cite{canetti1993fast}}.
Roughly speaking, an AVSS scheme consists of a {\tmem{sharing}} phase and a
{\tmem{reconstruction}} phase, involving a process designated as the
{\tmem{dealer}} which has a value (usually called {\tmem{secret}}) to share. In
the sharing phase, the dealer shares its secret among all processes and each
process locally verifies that a unique secret is being considered. In the
reconstruction phase, the processes reconstruct the secret from the shares.  The
correctness of AVSS lies on two properties:  (1) if the dealer is correct, then
all correct processes will reconstruct the secret of the dealer, and  (2) if
the dealer is faulty, then all correct processes will reconstruct the same
value that is fixed in the sharing phase.

We introduce in this paper a variant of AVSS called IVSS (standing for
{\tmem{inferable (asynchronous) verifiable secret sharing}}).
Our IVSS scheme has a weaker correctness property than AVSS, but provides strong
fault-detection ability. Specifically,
IVSS requires that if the correctness property of AVSS does not hold in an
invocation of some round, then correct processes will ignore (or infer) at
least $t (n - 3 t)$ faulty pairs from that round on. Here, by a \tmem{faulty pair},
we mean a pair of processes of which at least one is faulty.
In our IVSS protocol, secrets are shared through symmetric bivariate polynomials.
If processes reconstruct different secrets in the protocol, the symmetry of polynomials can be used to infer faulty pairs.

There are existing secret sharing protocols with fault-detection capacity, e.g., 
shunning verifiable secret sharing in \cite{beerliova2006efficient} and secret sharing with dispute control in \cite{abraham2008almost}.
These protocols are composed of several levels of secret sharing subprotocols, while our protocol is very simple with only one-level secret sharing subprotocol.
In all previous approaches, the Byzantine agreement algorithm proceeds
round by round and, once a round is over, the correct processes forget it and
never look back to it. In fact, if a correct process could look back at
the history of invocations of the secret sharing protocol, it may infer more
failures.
We implement this history-based checking mechanism in a {\tmem{certification}} subprotocol.
This subprotocol is invoked when the Byzantine agreement protocol is initialized
and then runs concurrently with all invocations of our IVSS protocol.
The main technique for inferring faults in our protocol is also different from \cite{beerliova2006efficient,abraham2008almost}.
Our fault-detection mechanism is based on symmetric polynomials which enable our protocol to infer a linear number of faults when secret sharing does not succeed, while protocols in \cite{beerliova2006efficient,abraham2008almost} can generally infer only one fault.

The rest of this paper is organized as follows. In Section \ref{sec-model}, we
recall the asynchronous computing model and the Byzantine agreement problem.
In Section \ref{sec-ivss}, we state the properties of our IVSS scheme and describe 
an algorithm that implements it.  In Section \ref{sec-aba}, we show how to obtain our 
fast ABA protocol from our IVSS scheme. For space limitations, some algorithms
and proofs are given in the appendices.

\section{Model and Definitions}\label{sec-model}

\subsubsection*{The Model}

We consider an  asynchronous computing model in the classical sense, e.g., 
{\cite{canetti1993fast,abraham2008almost}}. We consider a complete network of $n$
processes with identifiers $\{ 1 \nocomma, 2, \ldots, n \}$.  The number $n$ is always
strictly greater than $3 t$. The communication channels are
private, i.e. no one can read or alter messages transmitted along it.
Messages sent on a channel may have arbitrary (but finite) delay.
A $t$-adversary can control at most $t$ processes during the Byzantine agreement protocol. 
Once a process is controlled, it hands all its data over
to the adversary and follows its instructions. We
call all these controlled processes as {\tmem{faulty}} ones and other uncontrolled
processes as {\tmem{correct}} ones.
Note that the adversary cannot access messages transmitted between correct processes due to private communication channels.

We measure the running time of a protocol by the maximal expected number of communication rounds it takes to reach agreement {\cite{canetti1993fast, king2014faster}}.
Consider a virtual `global clock'
measuring time in the network. This clock cannot be accessed by the processes.
Let the {\tmem{delay}} of a message transmission denote the time elapsed from its sending to its
reception. The {\tmem{period}} of a finite execution of a protocol is the
longest delay of a message transmission during this execution.
Let the {\tmem{duration}} of a finite
execution denote the total time measured by the global clock divided by the period
of this execution. The {\tmem{expected running time}} of a protocol,
is the maximum over all inputs and applicable adversaries, of the average of the
duration of executions of the protocol over the random inputs of the processes.
In addition, each process divides its local time into {\tmem{rounds}} and execute a protocol
round by round.
The time of each round is less than or equal to a period of the execution of a protocol.
The expected running time of a protocol can be computed by the expected rounds in execution.

\subsubsection*{Asynchronous Byzantine Agreement}

\begin{definition}[ABA]
  Let $\pi$ be any asynchronous protocol in which each process has a binary
  input. We say that $\pi$ is an almost-surely terminating, $t$-resilient ABA
  protocol if the following properties hold for every $t$-adversary and every
  input:
  \begin{itemize}
    \item {\tmstrong{Termination}}: With probability one, every correct
    process terminates and outputs a value.
    
    \item {\tmstrong{Correctness}}: All correct processes which have
    terminated have the same outputs. Moreover, if all correct processes
    have the same input, denoted $v$, then all correct processes output
    $v$.
  \end{itemize}
\end{definition}

\subsubsection*{Asynchronous Broadcast: A-Cast}

{\noindent} We will often make use of this asynchronous broadcast primitive,
introduced by Bracha {\cite{bracha1984asynchronous}} (for $n > 3 t$).  We
follow the terminology in {\cite{canetti1996studies}}. For completeness, the
implementation is provided in Appendix \ref{app-acast}.

\begin{definition}[A-Cast]
  Let $\pi$ be any asynchronous protocol initiated by a designated process (the
  sender) which has an input value $u$ to be broadcast. We say that $\pi$ is a
  $t$-resilient A-Cast protocol if the following properties hold for every
  $t$-adversary:
  \begin{itemize}
    \item {\tmstrong{Termination}}:
    \begin{enumeratenumeric}
      \item If the sender is correct and all  correct processes
      participate in $\pi$, then every correct process eventually
      completes $\pi$.
      \item If some correct process completes $\pi$, then every
      correct process eventually completes $\pi$.
    \end{enumeratenumeric}
    \item {\tmstrong{Correctness}}:
    \begin{enumeratenumeric}
      \item All correct processes which complete $\pi$ receive the same
      value $v$.
      
      \item If the sender is correct, then $v = u$.
    \end{enumeratenumeric}
  \end{itemize}
\end{definition}

\section{Inferable Verifiable Secret Sharing}\label{sec-ivss}

In this section, we first state the properties of our IVSS scheme. Then we provide
an implementation of IVSS. We prove that our implementation satisfies all the
IVSS properties and finally we analyze its fault-detection.

\subsection{Definition}

\begin{definition}[Faulty Pair]
  An unordered pair $\{ i, j \}$ of processes is called a {\tmem{faulty pair}} if either $i$ or $j$ is faulty.
\end{definition}

Our IVSS protocol consists of two subprotocols: $\mathcal{S}$ ({\tmem{sharing}}
protocol) and $\mathcal{R}$ ({\tmem{reconstruction}} protocol). These two
 are invoked separately but $\mathcal{R}$ is never called unless
$\mathcal{S}$ is completed, and $\mathcal{R}$ may not be called even if
$\mathcal{S}$ is completed. If
the correct processes do not reconstruct a same secret in $\mathcal{R}$, then a set of faulty
pairs will be inferred. We assume that each
IVSS invocation is unique for every correct process. This can be easily
guaranteed, e.g. by associating with each IVSS invocation the identifier of
the dealer and an invocation counter.

\begin{definition}[IVSS]
  Let $(\mathcal{S}, \mathcal{R})$ be any pair of sharing-reconstruction protocol
  with a dealer which has a secret $s$ to share. We say that $(\mathcal{S},
  \mathcal{R})$ is an {\tmname{$\tmop{IVSS}$}} protocol if the following
  properties (called \tmem{IVSS properties}) hold.
  \begin{itemize}
    \item {\tmstrong{Termination}}:
    \begin{enumeratenumeric}
      \item If the dealer is correct and all correct processes keep
      participating in protocol $\mathcal{S}$, then every correct process eventually
      completes protocol $\mathcal{S}$.
      
      \item If some correct process completes protocol $\mathcal{S}$, then
      every correct process that keeps participating in protocol $\mathcal{S}$
      eventually completes protocol $\mathcal{S}$.
      
      \item If some correct process completes protocol $\mathcal{S}$ and all
      correct processes begin protocol $\mathcal{R}$ and keep participating in
      protocol $\mathcal{R}$, then every correct process eventually completes
      protocol $\mathcal{R}$.
      
      \item If some correct process completes protocol $\mathcal{R}$, then
      every correct process that keeps participating in protocol $\mathcal{R}$
      eventually completes protocol $\mathcal{R}$.
    \end{enumeratenumeric}
    \item {\tmstrong{Correctness}}: Once a correct process has completed
    protocol $\mathcal{S}$, then there is a unique value $v$ such that the
    following holds.
    \begin{enumeratenumeric}
      \item Either every correct process upon completing protocol $\mathcal{R}$
      outputs $v$, or a set of new faulty pairs is eventually  inferred by
      correct processes. (In our implementation, the size of the set of new faulty pairs is at least $t (n - 3t)$.)
      
      \item If the dealer is correct, then $v = s$.
    \end{enumeratenumeric}
    \item {\tmstrong{Secrecy}}: If the dealer is correct and no correct
    process invokes protocol $\mathcal{R}$, then the faulty processes have no
    information about secret $s$.
  \end{itemize}
\end{definition}

Note that, a correct process is said to keep participating in a protocol if it
follows the protocol until completion. Another note is that we assume all
secrets, random values, and polynomials to be over the integer ring.

\subsection{Implementation}

In our ABA protocol, the processes invoke a set of secret sharing instances in
each round (starting from round $1$). Every process records its invocations in
each round $r$ and A-Casts these invocations in the next round $r + 1$ to let other processes
know about its behavior in round $r$. We introduce a new component, which we call the
{\tmem{certification protocol}}, to take care of the IVSS invocations from past
rounds and infer faulty pairs. The certification protocol is invoked before
round $1$ and runs concurrently with all invocations of IVSS. Hence
our IVSS protocol should be aware of the particular round it is involved in, and
should make progress based on the data from past rounds. Therefore, we use the
notion IVSS[$r$] with round number $r$ as a parameter. In this section,
we give a high-level description of our IVSS[$r$] and our certification protocols.

\begin{figure}[thp]
\begin{center}
\setlength\doublerulesep{1.5mm} 
\begin{tabular}{|p{0.97\textwidth}|}
\hline
  {{\tmem{{\tmstrong{Sharing protocol}}}} IVSS[$r$]-$\mathcal{S}$:} \\
\hline
  \vspace{-3mm}
  \begin{enumeratenumeric}
    \item If the dealer wants to share secret $s$ in round $r$, it selects a
    random degree-$t$ symmetric bivariate polynomial $f (x, y)$ such that $f
    (0, 0) = s$. Let $f_i$ denote the degree-$t$ polynomial such that $f_i
    (y) = f (i, y)$ for $y \in \{ 1, \ldots, n \}$. The dealer sends
    $f_i$ to process $i$.
    
    \item If process $k$ receives $\widehat{f_k}$ from the dealer, then $k$
    sends $\widehat{f_k} (i)$ to process $i$. (Note that $\widehat{f_k}$ is
    supposed to be $f_k$ if the dealer is correct.)
    
    \item If process $k$ receives $\widehat{f_k}$ from the dealer and receives
    $\widehat{f_i (k)}$ from process $i$, and $\widehat{f_k} (i) =
    \widehat{f_i (k)}$, then $k$ A-Casts ``{\tmsamp{equal: }}$(k, i)$''. (Note
    that $\widehat{f_i (k)}$ is supposed to be $\widehat{f_i} (k)$ if $i$ is
    correct.)
    
    \item \label{mr-ivss-s-decideM}If there is a set $\mathcal{M}$ of $n - t$
    processes such that the following conditions are satisfied for the dealer:
    \begin{enumeratealpha}
      \item for every $i, j \in \mathcal{M}$, the dealer receives
      ``{\tmsamp{equal: }}$(i, j)$'';
      
      \item for every $i, j, p, q \in \mathcal{M}$, the dealer receives
      ``{\tmsamp{checked}}$_r \nocomma : p, q, \{ i, j \}$'' from $p$,
    \end{enumeratealpha}
    then the dealer A-Casts $\mathcal{M}$. ($\mathcal{M}$ is called
    {\tmem{candidate set}}.)
    
    \item \label{mr-ivss-s-complete}If process $k$ receives $\mathcal{M}$ from
    the dealer and the following conditions are satisfied:
    \begin{enumeratealpha}
      \item for every $i, j \in \mathcal{M}$, $k$ receives ``{\tmsamp{equal:
      }}$(i, j)$'';
      
      \item for every $i, j, p, q \in \mathcal{M}$, $k$ receives
      ``{\tmsamp{checked}}$_r \nocomma : p, q, \{ i, j \}$'' from $p$,
    \end{enumeratealpha}
    then $k$ completes the sharing protocol.
  \end{enumeratenumeric} \vspace{-5mm}\\
\hline
\hline
  {{\tmem{{\tmstrong{Reconstruction protocol}}}} IVSS[$r$]-$\mathcal{R}$:} \\
\hline
  \vspace{-3mm}
  \begin{enumeratenumeric}
    \item \label{ivss-r-rbf}If process $k \in \mathcal{M}$, then $k$ A-Casts
    polynomial $\widehat{f_k}$.
    
    \item \label{mr-ivss-r-ready}If there is a set $I S_k$ (standing for
    Interpolation Set) of $n - 2 t$ processes such that
    \begin{enumeratealpha}
      \item $k$ receives $\widetilde{f_i}$ from each process $i \in I S_k$;
      (Note that $\widetilde{f_i}$ is supposed to be $\widehat{f_i}$ if $i$ is
      correct.)
      
      \item there is a symmetric bivariate degree-$t$ polynomial $\bar{f}$
      such that $\bar{f} (i, j) = \widetilde{f_i} (j)$ for all $i \in I S_k$
      and $j \in \mathcal{M}$,
    \end{enumeratealpha}
    then $k$ sets $v = \bar{f} (0, 0)$, A-Casts ``{\tmsamp{ready to
    complete}}'' and adds this instance of IVSS[$r$] to
    $\tmop{CoreInvocations}^k_r$.
    
    \item \label{mr-ivss-r-complete}If $k$ completes Step
    \ref{mr-ivss-r-ready} and receives ``{\tmsamp{ready to complete}}'' from
    $n - t$ processes, then $k$ outputs $v$ and completes the reconstruction
    protocol.
  \end{enumeratenumeric} \vspace{-5mm}\\
\hline
\hline
  {{\tmem{{\tmstrong{Certification protocol}}}}:} \\
\hline
  \vspace{-3mm}
  \begin{enumeratenumeric}
    \item Process $k$ initializes empty sets $F P_k$ and $\tmop{CoreInvocations}^k_0$.
    
    \item Process $k$ sets $\tmop{CoreInvocations}_r^k = \varnothing$ and
    A-Casts $\tmop{CoreInvocations}_{r - 1}^k$ in the beginning of round $r$ ($r \geqslant 1$).

    \item \label{mr-ivss-infer}({\tmem{Infer faulty pairs}}) If $k$ receives
    $\tmop{CoreInvocations}^l_r$ from process $l$, then for any instance
    $\mathbbm{I}$ in $\tmop{CoreInvocations}^l_r$, if $k$ receives
    $\widetilde{f_i}$ and $\widetilde{f_j}$ from process $i$ and $j$ ($i, j
    \in \mathcal{M}$ of $\mathbbm{I}$) in Step \ref{ivss-r-rbf} of
    IVSS-$\mathcal{R}$ such that $\widetilde{f_i} (j) \neq \widetilde{f_j} (i)$,
    then $k$ adds unordered pair $\{ i, j \}$ to $F P_k$.
    
    \item \label{coreInvo-finishLine1}If $k$ receives
    $\tmop{CoreInvocations}^l_r$ from process $l$, then for any invocation
    $\mathbbm{I}$ in $\tmop{CoreInvocations}^l_r$, $k$ completes the sharing
    protocol of $\mathbbm{I}$ and Step \ref{ivss-r-rbf} of IVSS[$r$]-$\mathcal{R}$ of $\mathbbm{I}$.
    (Note that $k$ does this because different process might complete different
    instances of IVSS[$r$] in round $r$.)

    \item \label{history-checkPair}If the following conditions are satisfied
    for process $k$ (check in order a, b, c):
    \begin{enumeratealpha}
      \item $k$ receives $\tmop{CoreInvocations}^l_{r'}$ from process $l$ for
      all $r' < r$;
      
      \item for every IVSS invocation $\mathbbm{I}$ in $\underset{\nocomma r'
      < r}{\cup} \tmop{CoreInvocations}^l_{r'}$, if $i$ ($j$ resp.) is
      included in the candidate set $\mathcal{M}$ of $\mathbbm{I}$ then $k$
      should receive the polynomial A-Cast by $i$ ($j$ resp.) in Step
      \ref{ivss-r-rbf} of IVSS-$\mathcal{R}$ of $\mathbbm{I}$;
      
      \item $\{ i, j \} \nin F P_k$ (Here $F P_k$ has been updated after
      checking condition b, see Step \ref{mr-ivss-infer}),
    \end{enumeratealpha}
    then $k$ A-Casts ``{\tmsamp{checked}}$_r \nocomma ; k, l, \{ i, j \}$''.
    (Intuitively, this means $k$ has checked that $\{ i, j \}$ is not a faulty
    pair according to the invocation history of $l$ before round $r$.)
  \end{enumeratenumeric} \vspace{-5mm}\\
\hline
\end{tabular}
\end{center}
\end{figure}

In the sharing phase, we assume that the dealer with secret $s$ selects a random
{\tmem{degree-$t$ symmetric bivariate polynomial}} $f$ such that $f (0, 0) =
s$. Let $f_i$ denote the degree-$t$ polynomial such that $f_i (y) = f (i,
y)$ for $y \in \{ 1, \ldots, n \}$. The dealer shares secret $s$ by
sending polynomial $f_i$ to process $i$. By polynomial interpolation, if
the dealer is correct, then any $t + 1$ correct processes could reconstruct
$f$. Since $f$ is a symmetric polynomial, we should have $f_i (j) = f_j (i)$.
Each process $k$ that receives $f_k$ sends $f_k (i)$ to process $i$. When $k$
receives $f_i (k)$ from process $i$, $k$ checks whether $f_i (k) = f_k
(i)$. This equality may not be true since the dealer or process $i$ could be
faulty. If the equality is correct, then $k$ A-Casts ``{\tmsamp{equal: }}$(k,
i)$''. When the dealer receives ``{\tmsamp{equal: }}{$i$}''
from every process $i$ in a set $\mathcal{M}$ that contains $n - t$ processes, and checks that
$\mathcal{M}$ does not contains faulty pairs according to the IVSS invocations
in the past rounds (see the description of the certification protocol below), the dealer
A-Casts $\mathcal{M}$. Intuitively, $\mathcal{M}$ is a candidate set that
processes could trust to reconstruct the secret. If process $k$ receives
set $\mathcal{M}$ from the dealer and checks the correctness of
$\mathcal{M}$ as the dealer, then $k$ completes the sharing protocol.

In the reconstruction phase, processes in $\mathcal{M}$ A-Cast their polynomials
received from the dealer. When process $k$ receives polynomials from $n-2t$ processes
and these polynomials can be interpolated to a degree-$t$ symmetric
bivariate polynomial $\bar{f}$, $k$ considers $\bar{f} (0, 0)$ as
the dealer's secret. If $\bar{f}$ is not equal to the polynomial $f$ selected
by the dealer in the sharing phase, we can show that a set of faulty
pairs will be inferred. In order to get every secret sharing instance checked by
the certification protocol in the next round, it is important that, when
a correct process completes a secret sharing invocation, at least $t + 1$
correct processes take this invocation as its history invocation. 
Therefore, after getting polynomial $\bar{f}$, $k$ first A-Casts
a message ``{\tmsamp{ready to complete}}'' and records the invocation. Then $k$
completes the reconstruction phase if $k$ receives $n - t$ ``{\tmsamp{ready to
complete}}''.

Our certification protocol handles the  history of invocations.  Process
$k$ uses set $F P_k$ to track the faulty pairs it inferred. In each round
$r$, $k$ records all invocations of IVSS[$r$] and adds them into a set
called $\tmop{CoreInvocations}^k_r$. Then, at the beginning of round $r + 1$,
$k$ will A-Cast $\tmop{CoreInvoations}^k_r$ to let other processes know its
action in round $r$. Intuitively, this means that every correct process should know
what the other processes have done in the past rounds.  If a process $k$
receives $f_i$ from $i$ and $f_j$ from $j$ but $f_i (j) \neq f_j (i)$ for some
IVSS instance, then $k$ knows that at least one of $i$, $j$ is faulty and adds
unordered pair $\{ i, j \}$ into $F P_k$. The word ``inferable'' in IVSS means
that correct pairs could infer faulty pairs during the execution.  If $k$
receives $\tmop{CoreInvoations}^l_r$ from process $l$, then it checks for each
invocation $\mathbbm{I}$ in $\tmop{CoreInvoations}^l_r$ that every correct
process in $\mathcal{M}$ of $\mathbbm{I}$ should A-Cast its polynomial in the
beginning of the reconstruction phase, and no pair of correct processes should
be considered as a faulty pair
according to these invocations. If $k$ has checked that an
unordered pair $\{ i, j \}$ is not a faulty pair according to the
invocation history of $l$ before round $r$, then $k$ will A-Cast
``{\tmsamp{checked}}$_r \nocomma : k, l, \{ i, j \}$''. In the sharing protocol,
a correct process accepts a candidate set $\mathcal{M}$ only if every pair
of processes in $\mathcal{M}$ are checked by every process in $\mathcal{M}$.

\subsection{Proof of IVSS properties}

\begin{lemma}
  \label{lemma-fp-correct}If $i, j, k$ are correct processes,
  then unordered pair $\{ i, j \}$ will not be added to $F P_k$.
\end{lemma}

\begin{proof}
  The pair $\{ i, j \}$ will be added to $F P_k$ only if there is an invocation $\mathbbm{I}$
  of IVSS[$r$] such that
  $i, j \in \mathcal{M}$ and the polynomials $\widetilde{f_i}$ and
  $\widetilde{f_j}$ A-Casted by $i$ and $j$ in Step \ref{ivss-r-rbf} of
  IVSS[$r$]-$\mathcal{R}$ satisfy $\widetilde{f_i} (j) \neq
  \widetilde{f_j} (i)$. However, if $i, j \in \mathcal{M}$ then $i$ and $j$
  must have A-Casted ``{\tmsamp{equal: }}$(i, j)$'' and ``{\tmsamp{equal: }}$(j,
  i)$'' and hence must have checked that $\widetilde{f_i} (j) = \widetilde{f_j} (i)$ in
  IVSS[$r$]-$\mathcal{S}$. Thus $\{ i, j \}$ will not be added to $F P_k$.
\end{proof}

\begin{lemma}
  \label{lemma-checked}In round $r$ ($r \geqslant 1$), if $i$, $j$, $k$, and $l$ are correct
  processes, then $k$ eventually A-Cast ``{\tmsamp{checked}}$_r \nocomma
  : k, l, \{ i, j \}$''.
\end{lemma}

\begin{proof}
  Since $\{ i, j \}$ is not in $F P_k$ by Lemma
  \ref{lemma-fp-correct}, we only need to check conditions {\tmstrong{a}} and {\tmstrong{b}} of Step
  \ref{history-checkPair} in the certification protocol.
  
  Condition {\tmstrong{a}}: Since $l$ is correct, $l$ will A-Cast
  $\tmop{CoreInvocations}^l_r$ in the beginning of round $r$. Then $k$ will
  receive these $\tmop{CoreInvocations}^l_r$ by the correctness property of A-Cast.
  
  Condition {\tmstrong{b}}: Suppose that $i$ is in the set
  $\mathcal{M}$ of an IVSS[$r'$] invocation $\mathbbm{I}$ in
  $\underset{\nocomma r' < r}{\cup} \tmop{CoreInvocations}^l_{r'}$. Since $l$
  adds $\mathbbm{I}$ into its CoreInvocations, $l$ must have completed the
  sharing protocol of $\mathbbm{I}$. Then $i$ must have received polynomial
  $\widehat{f_i}$ from the dealer in invocation $\mathbbm{I}$. According
  to Step \ref{coreInvo-finishLine1} of the certification protocol, $i$ will complete
  Step \ref{ivss-r-rbf} of IVSS[$r'$]-$\mathcal{R}$ of $\mathbbm{I}$. So $k$
  will receive the polynomial A-Casted by $i$ in Step \ref{ivss-r-rbf} of
  IVSS[$r$]-$\mathcal{R}$ of $\mathbbm{I}$.
  
  Taking above together, $k$ will A-Cast ``{\tmsamp{checked}}$_r \nocomma : k,
  l, \{ i, j \}$''.
\end{proof}

\begin{lemma}
  \label{lemma-poly}Let $N$ be a subset of $\{ 1, \ldots, n \}$ and $| N |
  \geqslant t + 1$. Let $\{ f_i \}_{i \in N}$ be a set of degree-t univariate
  polynomials. If $f_i (j) = f_j (i)$ for all $i, j \in N$, then there is a
  unique symmetric bivariate degree-$t$ polynomial $f$ such that $f (i, j) =
  f_i (j)$ for all $i, j \in N$.
\end{lemma}

\begin{proof}
  Select any subset $N_0$ of $N$ such that $| N_0 | = t + 1$. Let
  \[ f_0 (x, y) = \sum_{i \in N_0 ; j \in N_0} \frac{\underset{k \in N_0, k
     \neq i}{\prod} (x - k) \underset{k \in N_0, k \neq j}{\prod} (y -
     k)}{\underset{k \in N_0, k \neq i}{\prod} (i - k) \underset{k \in N_0, k
     \neq j}{\prod} (j - k)} f_i (j) . \]
  By Lagrange interpolation, $f_0 (i, j) = f_i (j)$ for all $i, j \in N_0$.
  Since $f_i (j) = f_j (i)$, $f_0$ is a symmetric bivariate degree-$t$
  polynomial by definition. Now we prove that $f_0 (i, j) = f_i (j)$ for all
  $i, j \in N$.
  
  Consider any arbitrary $i$ in $N$. We have $f_i (j) = f_j (i) = f_0 (j, i)$
  for all $j \in N_0$. Since $f_0$ is symmetric, we have $f_i (j) = f_0 (i,
  j)$ for all $j \in N_0$. Since $| N_0 | = t + 1$, we have $f_i (y) = f_0 (i,
  y)$ for any $y$. Especially, we have $f_i (j) = f_0 (i, j)$ for all $j \in N$. Hence,
  $f_0$ satisfies $f_0 (i, j) = f (i, j)$. The uniqueness follows easily from
  Lagrange interpolation.
\end{proof}

\begin{theorem}
  Assume $n > 3 t$. Then the pair (IVSS[$r$]-$\mathcal{S}$,
  IVSS[$r$]-$\mathcal{R}$) satisfies all the IVSS properties.
\end{theorem}

\begin{proof}
  We check below the IVSS properties.
  
  {\tmstrong{Termination (1)}}: Suppose the dealer is correct and all correct
  processes keep participating in IVSS[$r$]-$\mathcal{S}$. Every
  correct process will receive correct messages from the dealer. Then for
  each pair $(i, j)$ of correct processes, $i$ will A-Cast ``{\tmsamp{equal:
  }}$(i, j)$''. By Lemma \ref{lemma-checked},
  for correct processes $i, j, k, l$, ``{\tmsamp{checked}}$_r \nocomma :
  k, l, \{ i, j \}$'' will be A-Cast by $k$. Thus the set of correct
  processes will satisfy the conditions in Step \ref{mr-ivss-s-decideM} of
  IVSS[$r$]-$\mathcal{S}$. Therefore, a correct dealer will A-Cast a set
  $\mathcal{M}$ with respect to Step \ref{mr-ivss-s-decideM} of
  IVSS[$r$]-$\mathcal{S}$. Since all messages that the dealer received in
  Step \ref{mr-ivss-s-decideM} are sent using A-Cast, it follows that all
  correct processes will receive $\mathcal{M}$ and check that $\mathcal{M}$
  satisfies the conditions in Step \ref{mr-ivss-s-complete} of
  IVSS[$r$]-$\mathcal{S}$. Hence, every correct process will complete
  IVSS[$r$]-$\mathcal{S}$.
  
  {\tmstrong{Termination (2)}}: If a correct process completes
  IVSS[$r$]-$\mathcal{S}$, then, since all messages required in Step \ref{mr-ivss-s-complete}
  of IVSS[$r$]-$\mathcal{S}$ are sent by A-Casting, every correct process that keeps
  participating in IVSS[$r$]-$\mathcal{S}$ will receive these messages and complete
  IVSS[$r$]-$\mathcal{S}$.

  {\tmstrong{Termination (3)}}: If some correct process completes protocol IVSS[$r$]-$\mathcal{S}$
  and all correct processes
  begin IVSS[$r$]-$\mathcal{R}$ and keep participating, we show that every correct process
  will complete IVSS[$r$]-$\mathcal{R}$. Let $C$ be the set of correct
  processes in $\mathcal{M}$. Since $| \mathcal{M} | \geqslant n - t$, then $| C |
  \geqslant n - 2 t$. Let $\widehat{f_i}$ be the polynomial $i$ ($i \in C$)
  received from the dealer. Since $C \subset \mathcal{M}$, we have
  $\widehat{f_i} (j) = \widehat{f_j} (i)$ for all $i, j \in C$. By Lemma
  \ref{lemma-poly}, there is a symmetric bivariate degree-$t$ polynomial
  $\bar{f}$ such that $\bar{f} (i, j) = \widehat{f_i} (j)$ for all $i, j \in
  C$. Thus $C$ satisfies the conditions in Step \ref{mr-ivss-r-ready} of
  IVSS[$r$]-$\mathcal{R}$. It follows that every correct process will
  complete Step \ref{mr-ivss-r-ready} of IVSS[$r$]-$\mathcal{R}$ and A-Casts
  ``{\tmsamp{ready to complete}}''. Therefore, every correct process will
  receive at least $n - t$ ``{\tmsamp{ready to complete}}'' messages and
  complete IVSS[$r$]-$\mathcal{R}$.

  {\tmstrong{Termination (4)}}: If a correct process completes
  IVSS[$r$]-$\mathcal{R}$, then, since all messages required for completing
  IVSS[$r$]-$\mathcal{R}$ are sent by A-Casting, every
  correct process that keeps participating in IVSS[$r$]-$\mathcal{R}$ will
  receive these messages and complete IVSS[$r$]-$\mathcal{R}$.
  
  We now turn to the correctness properties.
  
  Suppose that a correct process has completed the sharing protocol. By Lemma
  \ref{lemma-poly}, there is a symmetric bivariate degree-$t$ polynomial
  $\bar{f}$ such that $\bar{f} (i, j) = \widehat{f_i} (j)$ for all $i, j \in
  C$ where $C$ is the set of all correct processes in $\mathcal{M}$. We
  denote $\bar{f} (0, 0)$ as $v$.
  
  {\tmstrong{Correctness (1)}}: If some correct process $k$ completes
  IVSS[$r$]-$\mathcal{R}$ and outputs a value different from $v$, then $I S_k$
  must be different from $C$. And there must be some process $i \in I S_k$ and
  some process $j \in C$ such that $\bar{f} (i, j) \neq \widehat{f_i} (j)$,
  otherwise $I S_k$ also interpolates $\bar{f}$ and output
  $\bar{f} (0, 0)$. Since $\bar{f} (i, j) = \bar{f} (j, i) = \widehat{f_j}
  (i)$, we have $\widehat{f_i} (j) \neq \widehat{f_j} (i)$, which means some
  faulty pair will be inferred (we will analysis how many pairs could be inferred
  in the following section).
  
  {\tmstrong{Correctness (2)}}: If the dealer is correct, then
  $\widehat{f_i} (j) = f (i, j)$ for all $i, j \in C$ where $f$ is the
  polynomial selected by the dealer. Thus $\bar{f} = f$ and $v = \bar{f} (0,
  0) = f (0, 0) = s$.
  
  {\tmstrong{Secrecy}}: By polynomial interpolation, the combined view of the $t$ faulty processes is not
  enough to compute the initial random degree-$t$ polynomial selected by the
  dealer. As long as no correct process invokes IVSS[$r$]-$\mathcal{R}$, the
  shared secret is independent of the information obtained by the faulty
  processes. Hence, the faulty processes have no information of the shared
  secret.
  
  So all the IVSS properties hold for IVSS[$r$]. The theorem follows.
\end{proof}

\subsection{Fault-Detection Analysis}

We introduce the following convention for the analysis of Fault-Detection in the certification protocol.
Consider an instance $\mathbbm{R}$ of IVSS[$r$]-$\mathcal{R}$ in
$\tmop{CoreInvocations}^i_r$ for a correct process $i$. If faulty
process $l$ in $\mathcal{M}$ of $\mathbbm{R}$ does not send its polynomial in
Step $1$ of $\mathbbm{R}$, then in round $r'$ (greater than $r$),
no correct process will allow $l$ to appear in $\mathcal{M}$ of
IVSS[$r'$] (see condition (b) of Step \ref{mr-ivss-s-complete} in
IVSS[$r'$]-$\mathcal{S}$ and Step \ref{history-checkPair} of the certification protocol).
This is the best case for correct processes.
Therefore without loss of generality, we use the following convention.

{\smallskip}\tmstrong{Convention. }\tmem{In any instance of IVSS[$r$]-$\mathcal{R}$
and any round $r$, every faulty process in the corresponding set $\mathcal{M}$
eventually A-Casts a polynomial (can be arbitrary) according to Step
\ref{ivss-r-rbf} of IVSS[$r$]-$\mathcal{R}$.}{\hspace*{\fill}}{\smallskip}

Consider an arbitrary instance of IVSS[$r$].
With the above convention, let $\widehat{f_i}$ be the polynomial eventually
A-Casted by process $i \in \mathcal{M}$ in Step $1$ of IVSS[$r$]-$\mathcal{R}$. We
say that a set $S \subset \mathcal{M}$ of at least $n - 2 t$ processes is an
{\tmem{interpolation set}} if there is a symmetric bivariate degree-$t$
polynomial $g$ such that $g (i, j) = \widehat{f_i} (j)$ for all $i \in S$.
Two interpolation sets $S$ and $S'$ are different, if the corresponding
bivariate polynomial are different, which implies $| S \cap S' | \leqslant t$
by Lemma \ref{lemma-poly}.

For an instance $\mathbbm{I}$ of IVSS[$r$], recall that by Lemma
\ref{lemma-poly} the polynomials that processes in $C$ received from the dealer actually define
a unique symmetric bivariate degree-$t$ polynomial, and therefore define a
unique secret $s$. We say that $s$ is the secret defined by $\mathbbm{I}$.
\begin{definition}
Let $\mathbbm{E}$ be the event that at least one of the correct processes output
a value $s'$ in the reconstruction phase of $\mathbbm{I}$ such that $s' \neq s$.
\end{definition}
If $\mathbbm{E}$ never occurs, then we could get a common coin with
high probability (we will show this later in Section \ref{sec-aba}). Thus it
is significant to analyze the situation when $\mathbbm{E}$ occurs.

\begin{lemma}
  $\mathbbm{E}$ could only occur in some instance $\mathbbm{I}$ of
  IVSS[$r$] when $n \leqslant 4 t$.
\end{lemma}

\begin{proof}
  If $\mathbbm{E}$ occurs, then there are at least two different
  {\tmem{interpolation sets}}. One of these is the set $C$ of correct processes in
  $\mathcal{M}$, the other one is the interpolation set $I S$ causing some correct
  process to output a different secret. Since $| C | \geqslant n - 2 t, | I S | \geqslant n - 2 t, | C
  \cap I S | \leqslant t$, we have $| C \cup I S | = | C | + | I S | - | C
  \cap I S | \geqslant 2 n - 5 t$. If $n > 4 t$, then $| C \cup I S | > n -
  t$. This is impossible since $| C \cup I S | \leqslant | \mathcal{M} | = n - t$.
  Therefore, $\mathbbm{E}$ could only occur when $n \leqslant 4 t$.
\end{proof}

\begin{lemma}
  \label{lemma-infer}If $\mathbbm{E}$ occurs in some instance $\mathbbm{I}$ of
  IVSS[$r$], then at least $t (n - 3 t)$ faulty pairs will be
  inferred by every correct process due to $\mathbbm{I}$.
\end{lemma}

\begin{proof}
  When $\mathbbm{E}$ occurs, at least one correct process completes
  instance $\mathbbm{I}$. According to Step \ref{mr-ivss-r-complete} of
  IVSS[$r$]-$\mathcal{R}$, there are at least $n - 2 t$ correct processes that have
  A-Casted ``{\tmsamp{ready to complete}}''. According to Step
  \ref{mr-ivss-r-ready} of IVSS[$r$]-$\mathcal{R}$, these correct processes
  must have added $\mathbbm{I}$ into the set CoreInvocations. In the next
  round $r + 1$, the candidate set $\mathcal{M}$ of any instance of IVSS[$r +
  1$] will contain at least one of these $n - 2 t$ processes since
  $\mathcal{M}= n - t$ and $n - 2 t > t$. Thus $\mathbbm{I}$ will be checked
  by every correct process in the certification protocol. Since the faulty
  pairs are inferred from the polynomials A-Casted by processes in
  $\mathcal{M}$ of $\mathbbm{I}$, all correct processes will
  infer the same faulty pairs. So we only need to prove the lemma for
  correct process $k$.
  
  Let $\{ S_1, S_2 \ldots, S_r \}$ be all maximal interpolation sets with
  respect to the inclusion relation of sets. Since $\mathbbm{E}$ occurs, there must
  be at least two maximal interpolation sets, one of which implies the secret
  $s$ defined by $\mathbbm{I}$ and another of which implies the secret $s' \neq s$, i.e. $r \geqslant 2$.
  Suppose $i, j \in 1, \ldots, r$ and $i \neq j$. By the assumption of maximal
  interpolation sets, $| S_i \cap S_j | \leqslant t$. Let $S_0$ be the
  interpolation set in $\{ S_1, S_2 \ldots, S_r \}$ with the smallest cardinal
  number. Since $| S_i \cup S_j | \leqslant | \mathcal{M} | = n - t$ and $|
  S_i \cap S_j | \leqslant t$, then $| S_i | + | S_j | = | S_i \cup S_j | + | S_i
  \cap S_j | \leqslant n$. Therefore $| S_0 | \leqslant \frac{| S_i | + | S_j
  |}{2} \leqslant \frac{n}{2}$. Also from the definition of the interpolation set,
  we have $| S_0 | \geqslant n - 2 t$.
  
  Suppose the corresponding symmetric bivariate polynomial for $S_0$ is $f^0$.
  Let $f_i^0$ be the polynomial with $f_i^0 (j) = f^0 (i, j)$. Since $f^0$ is
  symmetric, $f_i^0 (j) = f^0 (i, j) = f^0 (j, i) = \widehat{f_j} (i)$ for
  every $j \in S_0$. Recall that $\widehat{f_j}$ is the polynomial
  eventually A-Casted by process $j \in \mathcal{M}$ in Step $1$ of
  IVSS[$r$]-$\mathcal{R}$ of instance $\mathbbm{I}$. For any $i \in
  \mathcal{M}$ but $i \nin S_0$, we have $\widehat{f_i} \neq f_i^0$ because
  otherwise $S_0 \cup i$ is an interpolation set bigger than $S_0$,
  which contradicts the fact that $S_0$ is maximal. Since $\widehat{f_i} \neq
  f_i^0$, $\widehat{f_i} - f_i^0$ has at most $t$ zero points. So there are at
  least $| S_0 | - t$ processes $j$ in $S_0$ such that $\widehat{f_i} (j) \neq
  f_i^0 (j)$, i.e. $\widehat{f_i} (j) \neq \widehat{f_j} (i)$ (since $f_i^0
  (j) = \widehat{f_j} (i)$ for $j \in S_0$) which leads to the faulty pair $\{ i,
  j \}$. Therefore, for each $i \in \mathcal{M}$ but $i \nin S_0$, $k$ will
  infer at least $| S_0 | - t$ faulty pairs. In total, $k$ could infer at
  least $(| S_0 | - t) (n - t - | S_0 |)$ faulty pairs. Since $n - 2 t
  \leqslant | S_0 | \leqslant \frac{n}{2}$, then $(| S_0 | - t) (n - t - | S_0 |)
  \geqslant (n - 3 t) t$. The lemma is proved.
\end{proof}

In the lemma above, we show that a set of faulty pairs will eventually be
inferred if $\mathbbm{E}$ occurs in an instance of IVSS[$r$]. However,
``eventually'' is not enough to improve running time. In the next lemma, we will
show that the faulty pairs inferred from instance of IVSS[$r$] will not appear
in candidate set $\mathcal{M}$ of IVSS[$r + 1$] even though these faulty
pairs might be inferred after the invocation of IVSS[$r + 1$].

\begin{lemma}
  \label{lemma-immediate}If $\mathbbm{E}$ occurs in some instance
  $\mathbbm{I}^r$ of IVSS[$r$], and $\{ i, j \}$ is eventually inferred as
  faulty pairs by the correct processes due to $\mathbbm{I}^r$, then $i$ and $j$
  could not appear simultaneously in the set $\mathcal{M}$ of any instance of
  IVSS[$r'$] with $r' > r$.
\end{lemma}

\begin{proof}
  Since $\mathbbm{E}$ occurs, there must be a correct process (say $k$) that
  completes instance $\mathbbm{I}^r$. Then, by Step
  \ref{mr-ivss-r-complete} of IVSS[$r$]-$\mathcal{R}$, $k$ must have received
  ``{\tmsamp{ready to complete}}'' from $n - t$ processes. According to Step
  \ref{mr-ivss-r-ready} of IVSS[$r$]-$\mathcal{R}$, these $n - t$ processes
  must have added instance $\mathbbm{I}^r$ into
  $\tmop{CoreInvocations}^{\ast}_r$. Then there are at least $n - 2 t$
  correct processes (denoted by $S$) that have added instance
  $\mathbbm{I}^r$ into $\tmop{CoreInvocations}_{\ast, r}$.
  
  Now consider round $r' > r$. In any instance $\mathbbm{I}^{r'}$ of
  IVSS[$r'$], set $\mathcal{M}$, A-Casted by the dealer, contains at least one
  correct process (denoted by $l$) from $S$ since $| \mathcal{M} | \geqslant n
  - t$ and $| S | \geqslant n - 2 t \geqslant t + 1$. If $i$ and $j$ are both
  in the set $\mathcal{M}$ of $\mathbbm{I}^{r'}$, then every correct process
  $k'$ in $\mathcal{M}$ must A-Cast ``{\tmsamp{checked}}$_{r'} \nocomma, k',
  l, \{ i, j \}$'' according to Step \ref{mr-ivss-s-complete} of
  IVSS[$r'$]-$\mathcal{S}$. By Step \ref{history-checkPair} of our
  certification protocol, $k'$ must have received the corresponding
  polynomials of $i$ and $j$ A-Cast in Step \ref{ivss-r-rbf} of
  IVSS[$r'$]-$\mathcal{R}$. However, this would make $k'$ add $\{ i, j \}$ into
  $F P_{k'}$ and not A-Cast ``{\tmsamp{checked}}$_{r'} \nocomma : k', l, \{ i,
  j \}$''. This is a contradiction. Therefore, $i$ and $j$ could not appear
  simultaneously in candidate set $\mathcal{M}$ of any instance of IVSS[$r'$] with
  $r' > r$.
\end{proof}

\begin{lemma}
  \label{lemma-error-rounds}
  If $n = 3t + \delta$, then there are at most $\frac{3 t}{\delta} + 1$ rounds
  where $\mathbbm{E}$ occurs.
\end{lemma}

\begin{proof}
  Suppose $\mathbbm{E}$ occurs in round $r_1, r_2, \ldots, r_c$ and denote the faulty
  pairs that could be inferred for these rounds by $S_1, S_2, \ldots, S_c$. By
  Lemma \ref{lemma-immediate}, $S_i$ is different from $S_j$ for $1
  \leqslant i, j \leqslant c$ and $i \neq j$. According to Lemma \ref{lemma-infer}, there will
  be at least $c \cdot (n - 3 t) t$ different faulty pairs inferred. Since each faulty process
  can only appear in $n$ faulty pairs, we have $t \cdot n \geqslant c \cdot (n
  - 3 t) t$. Thus, $c \leqslant \frac{t n}{(n - 3 t) t} = \frac{3 t}{\delta} +
  1$.
\end{proof}

\section{From IVSS to Asynchronous Byzantine Agreement}\label{sec-aba}

Using our IVSS[$r$] protocol, we now design an ABA protocol (following the
reduction scheme of Canetti and Rabin {\cite{canetti1993fast}}).
The first step is to get a common coin. In the common coin protocol of
{\cite{canetti1993fast}}, every process shares $n$ random
secrets using $n$ different invocations of the AVSS protocol of
{\cite{canetti1993fast}}. Following Figure 5-9 of {\cite{canetti1996studies}}
and using our IVSS[$r$] protocol, we obtain an {\tmem{Inferable Common Coin}} (ICC)
protocol which always terminates.

\begin{definition}[ICC]
  Let $\pi$ be any protocol where every process has a random input and a binary
  output. We say that $\pi$ is a terminating, $t$-resilient Inferable Common
  Coin protocol if the following properties (called \tmem{ICC properties}) hold for
  every $t$-adversary.
  \begin{itemize}
    \item {\tmstrong{Termination}}.
    \begin{enumeratenumeric}
      \item If all correct processes keep participating in $\pi$,
      then every correct process eventually completes.

      \item If some correct process completes $\pi$, then every other correct
      process that keeps participating in $\pi$ eventually completes.
    \end{enumeratenumeric}
    
    \item {\tmstrong{Correctness}}. For every invocation, either
    \begin{itemize}
      \item for each $v \in \{ 0, 1 \}$, with probability at least $1 /
      4$, every correct process upon completing $\pi$ outputs $v$; or
      
      \item a set of faulty pairs is eventually inferred by correct processes.
    \end{itemize}
  \end{itemize}
\end{definition}

\begin{lemma}
  For $n > 3 t$ and each round $r$, there is a terminating, $t$-resilient
  Inferable Common Coin protocol.
\end{lemma}

\begin{proof}
  The protocol implementing ICC by using our IVSS[$r$] subprotocol
  is a slight variant of figure 5-9 of {\cite{canetti1996studies}}. We call this
  protocol ICC[$r$]. The proof is in Appendix \ref{app-icc}.
\end{proof}

The second step is to use the common coin protocol to get an ABA protocol. In
{\cite{canetti1993fast}}, Canetti and Rabin use their
common coin protocol (that terminates with probability $1 - \varepsilon$) to
get an ABA protocol (that terminates with probability $1 - \varepsilon$). We
replace the common coin protocol of {\cite{canetti1993fast}} by ICC[$r$] to
obtain our almost-surely terminating ABA protocol.

\begin{theorem}[Byzantine Agreement]
  If $n = 3 t + \delta$, then there is an almost-surely terminating ABA
  protocol with expected running time $O ( \frac{t}{\delta} )$.
\end{theorem}

\begin{proof}
  By Lemma \ref{lemma-error-rounds}, we know there are at most $\frac{3
  t}{\delta} + 1$ rounds in which the adversary could break the correctness of secret
  sharing. In the rest of the rounds, all correct processes reconstruct
  the same value and this value is equal to the secret of the dealer if the
  dealer is correct, with which we can have a common coin that is sufficient for
  Byzantine agreement with constant expected running time. Therefore, the
  expected running time of our ABA protocol is $O ( \frac{t}{\delta} )$.
  we give the details in Appendix \ref{app-aba}.
\end{proof}

If we take $\delta = 1$ in the above theorem, we have the following corollary,
which improves the result of Abraham, Dolev, and Halpern
{\cite{abraham2008almost}}.

\begin{corollary}
  If $n = 3 t + 1$, then there is an almost-surely terminating, optimally
  resilient ABA protocol with expected running time $O (t)$.
\end{corollary}

If we take $\delta = \varepsilon t$ where $\varepsilon > 0$, we have the
following corollary, which improves the result of Feldman and Micali {\cite{feldman1988optimal}}.

\begin{corollary}
  If $n = (3 + \varepsilon) t$ where $\varepsilon > 0$, then there is an
  almost-surely terminating ABA protocol with expected running time $O (1 / \varepsilon)$.
\end{corollary}

\newpage

\bibliographystyle{acm}
\bibliography{fast}

\begin{thebibliography}{10}

\bibitem{abraham2008almost}
{\sc Abraham, I., Dolev, D., and Halpern, J.~Y.}
\newblock An almost-surely terminating polynomial protocol for asynchronous
  {B}yzantine agreement with optimal resilience.
\newblock In {\em Proceedings of the Twenty-Seventh ACM Symposium on Principles
  of Distributed Computing\/} (2008), PODC '08, ACM, pp.~405--414.

\bibitem{beerliova2006efficient}
{\sc Beerliov{\'a}-Trub{\'\i}niov{\'a}, Z., and Hirt, M.}
\newblock Efficient multi-party computation with dispute control.
\newblock In {\em Proceedings of the Third Conference on Theory of
  Cryptography\/} (2006), TCC '06, Springer-Verlag, pp.~305--328.

\bibitem{ben1983another}
{\sc Ben-Or, M.}
\newblock Another advantage of free choice (extended abstract): Completely
  asynchronous agreement protocols.
\newblock In {\em Proceedings of the Second Annual ACM Symposium on Principles
  of Distributed Computing\/} (1983), PODC '83, ACM, pp.~27--30.

\bibitem{bracha1984asynchronous}
{\sc Bracha, G.}
\newblock An asynchronous [(n-1)/3]-resilient consensus protocol.
\newblock In {\em Proceedings of the Third Annual ACM Symposium on Principles
  of Distributed Computing\/} (1984), PODC '84, ACM, pp.~154--162.

\bibitem{canetti1996studies}
{\sc Canetti, R.}
\newblock {\em Studies in secure multiparty computation and applications}.
\newblock PhD thesis, The Weizmann Institute of Science, 1996.

\bibitem{canetti1993fast}
{\sc Canetti, R., and Rabin, T.}
\newblock Fast asynchronous {B}yzantine agreement with optimal resilience.
\newblock In {\em Proceedings of the Twenty-fifth Annual ACM Symposium on
  Theory of Computing\/} (1993), STOC '93, ACM, pp.~42--51.

\bibitem{feldman1988optimal}
{\sc Feldman, P., and Micali, S.}
\newblock Optimal algorithms for byzantine agreement.
\newblock In {\em Proceedings of the Twentieth Annual ACM Symposium on Theory
  of Computing\/} (1988), STOC '88, ACM, pp.~148--161.

\bibitem{fischer1985impossibility}
{\sc Fischer, M.~J., Lynch, N.~A., and Paterson, M.~S.}
\newblock Impossibility of distributed consensus with one faulty process.
\newblock {\em J. ACM 32}, 2 (1985), 374--382.

\bibitem{karlin1986probabilistic}
{\sc Karlin, A., and Yao, A.}
\newblock Probabilistic lower bounds for {B}yzantine agreement.
\newblock {\em Unpublished document\/} (1986).

\bibitem{king2014faster}
{\sc King, V., and Saia, J.}
\newblock Faster agreement via a spectral method for detecting malicious
  behavior.
\newblock In {\em Proceedings of the Twenty-Fifth Annual ACM-SIAM Symposium on
  Discrete Algorithms\/} (2014), SODA '14, SIAM, pp.~785--800.

\bibitem{lamport1982byzantine}
{\sc Lamport, L., Shostak, R., and Pease, M.}
\newblock The {B}yzantine generals problem.
\newblock {\em ACM Trans. Program. Lang. Syst. 4}, 3 (1982), 382--401.

\bibitem{pease1980reaching}
{\sc Pease, M., Shostak, R., and Lamport, L.}
\newblock Reaching agreement in the presence of faults.
\newblock {\em J. ACM 27}, 2 (1980), 228--234.

\bibitem{rabin1983randomized}
{\sc Rabin, M.~O.}
\newblock Randomized {B}yzantine generals.
\newblock In {\em Proceedings of the Twenty-Fourth Annual Symposium on
  Foundations of Computer Science\/} (1983), FOCS '83, IEEE Computer Society,
  pp.~403--409.

\end{thebibliography}

\newpage
\normalsize
\appendices

\section{A-Cast Protocol}\label{app-acast}

\begin{definition}[A-Cast]
  Let $\pi$ be any asynchronous protocol initiated by a designated process (the
  sender) which has an input value $u$ to be broadcast. We say that $\pi$ is a
  $t$-resilient A-Cast protocol if the following properties hold for every
  $t$-adversary:
  \begin{itemize}
    \item {\tmstrong{Termination}}:
    \begin{enumeratenumeric}
      \item If the sender is correct and all  correct processes
      participate in $\pi$, then every correct process eventually
      completes $\pi$.
      \item If some correct process completes $\pi$, then every
      correct process eventually completes $\pi$.
    \end{enumeratenumeric}
    \item {\tmstrong{Correctness}}:
    \begin{enumeratenumeric}
      \item All correct processes which complete $\pi$ receive the same
      value $v$.
      
      \item If the sender is correct, then $v = u$.
    \end{enumeratenumeric}
  \end{itemize}
\end{definition}

\code{
  {{\tmem{{\tmstrong{A-Cast Protocol}}}}:}
}{
  \begin{enumeratenumeric}
    \item The sender with input $u$ sends ``{\tmsamp{msg}}: $u$'' to all
    processes.
    
    \item $i$ waits until receiving ``{\tmsamp{msg}}: $u$''. Then $i$ sends
    ``{\tmsamp{echo}}: $u$'' to all processes.
    
    \item $i$ waits until receiving $n - t$ ``{\tmsamp{echo}}: $u'$'' that
    agree on the value of $u'$. Then $i$ sends ``{\tmsamp{ready}}: $u'$'' to
    all processes.
    
    \item $i$ waits until receiving $t + 1$ ``{\tmsamp{ready}}: $u'$'' that
    agree on the value of $u'$. Then $i$ sends ``{\tmsamp{ready}}: $u'$'' to
    all processes.
    
    \item $i$ waits until receiving $2 t + 1$ ``{\tmsamp{ready}}: $u'$'' that
    agree on the value of $u'$. Then $i$ outputs $u'$ and completes the
    protocol.
  \end{enumeratenumeric}
}


\section{Inferable Common Coin Protocol}\label{app-icc}

\begin{definition}[ICC]
  Let $\pi$ be any protocol where every process has a random input and a binary
  output. We say that $\pi$ is a terminating, $t$-resilient Inferable Common
  Coin protocol if the following properties (called \tmem{ICC properties}) hold for
  every $t$-adversary.
  \begin{itemize}
    \item {\tmstrong{Termination}}.
    \begin{enumeratenumeric}
      \item If all correct processes keep participating in $\pi$,
      then every correct process eventually completes.

      \item If some correct process completes $\pi$, then every other correct
      process that keeps participating in $\pi$ eventually completes.
    \end{enumeratenumeric}
    
    \item {\tmstrong{Correctness}}. For every invocation, either
    \begin{itemize}
      \item for each $v \in \{ 0, 1 \}$, with probability at least $1 /
      4$, every correct process upon completing $\pi$ outputs $v$; or
      
      \item a set of faulty pairs is eventually inferred by correct processes.
    \end{itemize}
  \end{itemize}
\end{definition}

Our implementation (called ICC[$r$]) of ICC follows
{\cite{canetti1993fast}}. Roughly speaking, the protocol consists of two phases.
First, every process shares $n$ random secrets using our IVSS[$r$]-$\mathcal{S}$
protocol. The $i$th secret shared by each process is assigned to process $i$.
Once a process $i$ completes $t+1$ sharing protocols of secrets assigned to it,
$i$ A-Casts the identity of the dealers of these secrets. After this, by the
correctness property of IVSS[$r$], a fixed value (yet unknown) is \tmem{attached} to $i$.
The second phase is to select a subset of processes (say $H$) and reconstruct the
attached values of $H$. Different processes may choose different $H$ to
reconstruct secrets. However, if an instance of IVSS[$r$]-$\mathcal{R}$ is
invoked by a strict subset of correct processes, then there is no guarantee of termination.
Hence, in ICC[$r$] we require every process to A-Cast its $H$ before completion,
so that each process could try to reconstruct values with different $H$.

\code{
  {{\tmem{{\tmstrong{ICC[$r$] protocol}}}}: code for process $i$}
  }{
  \begin{enumeratenumeric}
    \item Choose a random value $x_{i, j}$ for all $1 \leqslant j \leqslant n$
    and invoke IVSS[$r$]-$\mathcal{S}$ as a dealer for this value. Denote this
    execution by IVSS[$r$]-$\mathcal{S}$($x_{i, j}$).
    
    \item Participate in IVSS[$r$]-$\mathcal{S}$($x_{j, k}$) for every $j, k
    \in \{ 1, \ldots, n \}$.
    
    \item Define a set $\mathcal{T}_i$. Add process $j$ to $\mathcal{T}_i$ if
    all IVSS[r]-$\mathcal{S}$($x_{j, l}$) have been completed for all $1
    \leqslant l \leqslant n$. Wait until $| \mathcal{T}_i | = t + 1$, then
    assign $T_i =\mathcal{T}_i$ and A-Cast ``{\tmsamp{attach}} $T_i$
    {\tmsamp{to}} $i$''. (we say that the secrets $\{ x_{j, i} | \nobracket j
    \in T_i \}$ are attached to process $i$.)
    
    \item Define a set $\mathcal{A}_i$. Add process $j$ to $\mathcal{A}_i$ if
    the A-Cast ``{\tmsamp{attach}} $T_j$ {\tmsamp{to}} $j$'' has been
    completed and $T_j \subseteq \mathcal{T}_i$. Wait until $| \mathcal{A}_i | =
    n - t$, then assign $A_i =\mathcal{A}_i$ and A-cast ``$i$
    {\tmsamp{accepts}} $A_i$''.
    
    \item Define a set $\mathcal{S}_i$. Add process $j$ to
    $\mathcal{S}_i$ if ``$j$ {\tmsamp{accepts}} $A_j$'' is received from $j$
    and $A_j \subseteq \mathcal{A}_i$. Wait until $| \mathcal{S}_i | = n - t$,
    then A-Cast ``{\tmsamp{Reconstruct Enabled}}''. Let $S_i$ denote the current
    content of $\mathcal{S}_i$ and $H_i$ denote the
    current content of $\mathcal{A}_i$. Then A-Cast $(H_i, S_i)$.
    
    \item Participates in IVSS[$r$]-$\mathcal{R}$($x_{k, j}$) for every $k \in
    T_j$ and $j \in \mathcal{A}_i$. Let $y_{k, j}$ be the corresponding
    output.
    
    \item Let $u = \lceil 0.87 n \rceil$. Every process $j \in \mathcal{A}_i$
    is {\tmsamp{associated}} with a value, say $v_j$, which is computed as
    follows: $v_j = \left( \sum_{k \in T_j} y_{k, j} \right) \tmop{mod} u$.
    
    \item Wait until receiving $(\widetilde{H_j}, \widetilde{S_j})$ from $j$
    with $\widetilde{H_j} \subseteq \mathcal{A}_i$ and $\widetilde{S_j} \subseteq \mathcal{S}_i$ and the
    values {\tmsamp{associated}} with all processes in
    $\widetilde{H_j}$ are computed. Now if there exists a process $k \in
    \widetilde{H_j}$ such that $v_k = 0$, then output $0$. Otherwise output
    $1$.
  \end{enumeratenumeric}
}

We now state and prove the following lemmas which are slight variants of
lemmas 5.28-5.31 presented in {\cite{canetti1996studies}}.

\begin{lemma}
  If some correct process completes ICC[$r$], then every other correct
  process that keeps participating in ICC[$r$] eventually completes.
\end{lemma}
\begin{proof}
  If a correct process $i$ completes ICC[$r$] with respect to $(\widetilde{H_j}, \widetilde{S_j})$,
  then, since all messages are sent by A-Casting, every
  correct process that keeps participating in ICC[$r$] will receive at least
  $t + 1$ $(\widetilde{H_j}, \widetilde{S_j})$ as well. By the termination property (4) of IVSS[$r$],
  every correct process that keeps participating will also compute the values {\tmsamp{associated}} with
  all the processes in $\widetilde{H_j}$ and then complete the protocol.
\end{proof}

\begin{lemma}
  \label{icc-termination}If all correct processes keep participating in
  ICC[$r$], then all correct processes complete ICC[$r$] in constant time.
\end{lemma}

\begin{proof}
  First we show that every correct process will A-Cast
  ``{\tmsamp{Reconstruct Enabled}}''. By termination property (1) of our
  IVSS[$r$] protocol, every correct process eventually completes
  IVSS-$\mathcal{S}$($x_{j, k}$) for every $k \in \{ 1, \ldots, n \}$ and correct $j$. Since
  there are at least $n - t$ correct processes, for each correct process
  $i$, $\mathcal{T}_i$ will eventually contain at least $t + 1$ (actually $n -
  t$) processes and thus $i$ will eventually A-Cast ``{\tmsamp{attach}} $T_i$
  {\tmsamp{to}} $i$''. So eventually, correct process $i$ will receive
  ``{\tmsamp{attach}} $T_j$ {\tmsamp{to}} $j$'' from every correct process
  $j$. Now since every process $k$ that is included in $\mathcal{T}_j$ will be
  eventually included in $\mathcal{T}_i$ (by termination property (2) of
  IVSS[$r$]), $T_j \subseteq \mathcal{T}_i$ will eventually hold. Therefore,
  every correct process $j$ will eventually be included in $\mathcal{A}_i$.
  Thus for every correct process $i$, $\mathcal{A}_i$ will eventually be of
  size $n - t$ and hence $i$ will A-Cast ``$i$ {\tmsamp{accepts}} $A_i$''.
  Following the same argument, $\mathcal{S}_i$ will be of size $n - t$ and
  hence $i$ will A-Cast ``{\tmsamp{Reconstruct Enabled}}'' and A-Cast $(H_i, S_i)$.
  
  We now show that all correct processes will complete ICC[$r$].
  By the lemma above, we only need to show that at least
  one of the correct processes will complete ICC[$r$]. Suppose by contradiction that no correct
  process will complete ICC[$r$]. Let $i$ be a correct process. If $i$
  receives ``{\tmsamp{attach}} $T_j$ {\tmsamp{to}} $j$'' from $j$ and includes
  $j$ in $\mathcal{A}_i$, then eventually every other correct process will
  do the same. Hence if $i$ invokes IVSS[$r$]-$\mathcal{R}$($x_{k, j}$) for $k \in
  T_j$ and $j \in \mathcal{A}_i$, then
  eventually every other correct process will also invoke
  IVSS[$r$]-$\mathcal{R}$($x_{k, j}$). By termination property (3) of
  IVSS[$r$], all correct processes will complete this
  IVSS[$r$]-$\mathcal{R}$($x_{k, j}$). Therefore,
  the values {\tmsamp{associated}} with all processes in
  $H_i$ will be computed.
  So $i$ will complete the protocol, in contradiction with the assumption that no correct
  process will complete.
  Therefore, all correct processes will complete.
  
  In the certification protocol, each process need to check all the past invocations in
  past rounds. However, since every correct process does this in every round, it is equal
  to that each process in every round checks all the invocations in the previous round. Therefore, all
  ``{\tmsamp{checked}}$_r \nocomma ; k, l, \{ i, j \}$'' could be finished in
  constant time for correct process $i, j, k, l$. Thus all invocations of IVSS[$r$]-$\mathcal{S}$ and
  IVSS[$r$]-$\mathcal{R}$ in ICC[$r$] complete in constant time. Since all
  A-Casts also complete in constant time, our ICC[$r$] protocol completes
  in constant time as well.
\end{proof}

\begin{lemma}
  \label{icc-associate}In ICC[$r$], once some correct process $j$ receives
  ``{\tmsamp{attach}} $T_i$ {\tmsamp{to}} $i$'' from the A-Cast of $i$, a
  unique value $v_i$ is fixed such that
  \begin{enumeratenumeric}
    \item Every correct process will {\tmsamp{associate}} $v_i$ with $i$ or
    a set of faulty pairs will eventually be inferred by correct processes.
    
    \item Value $v_i$ is distributed uniformly over $[0, \ldots, u - 1]$ and is
    independent of the values {\tmsamp{associated}} with the other processes.
  \end{enumeratenumeric}
\end{lemma}

\begin{proof}
  The correctness property of IVSS[$r$] ensures that for each $k \in T_i$ there
  is a fixed value $y_{k, i}$ such that all correct processes will output
  $y_{k, i}$ in IVSS[$r$]-$\mathcal{S}$($x_{k, i}$) or a set of faulty pairs
  will be inferred. Let $v_i = \left( \sum_{k \in T_i} y_{k, i} \right) \tmop{mod} u$,
  then every correct process will 
  {\tmsamp{associate}} $v_i$ with $i$ except that event $\mathbbm{E}$ occurs
  in some instances of IVSS[$r$], i.e., a set of faulty pairs will eventually
  be inferred by correct processes.
  
  It remains to show that $v_i$ is uniformly distributed over $[0, \ldots,
  u - 1]$, and is independent of the values associated with the other processes. A
  correct process starts reconstructing the secrets attached to process $i$
  only after it completes the ``{\tmsamp{attach}} $T_i$ {\tmsamp{to}} $i$''
  A-Cast. So the set $T_i$ is fixed before any correct process invokes
  IVSS[$r$]-$\mathcal{R}$($x_{k, i}$) for some process $k$. The secrecy
  property of IVSS[$r$] now ensures that, by the time the set $T_i$ is fixed,
  the adversary view of the invocations of IVSS[$r$]-$\mathcal{S}$($x_{k, i}$)
  where the dealers are correct is distributed independently of the shared
  values. Since $T_i$ contains at least one correct process and every
  correct process's shared secrets are uniformly distributed and mutually
  independent, the sum $v_i$ is uniformly and independently distributed over
  $[0, \ldots, u - 1]$. 
\end{proof}

\begin{lemma}
  \label{icc-m}Once a correct process A-Casts ``{\tmsamp{Reconstruct
  Enabled}}'', there is a set $M$ such that
  \begin{enumeratenumeric}
    \item For every process $j \in M$, some correct process has received
    ``{\tmsamp{attach}} $T_j$ {\tmsamp{to}} $j$'' from the A-Cast of $j$.
    
    \item If any correct process $k$ receives $(\widetilde{H_j}, \widetilde{S_j})$ from $j$
    with $\widetilde{H_j} \subseteq \mathcal{A}_k$ and $\widetilde{S_j} \subseteq \mathcal{S}_k$ and the
    values {\tmsamp{associated}} with all processes in
    $\widetilde{H_j}$ are computed, then $M \subseteq \widetilde{H_j}$.
    
    \item $| M | \geqslant \frac{n}{3}$.
  \end{enumeratenumeric}
\end{lemma}

\begin{proof}
  Let $i$ be the first correct process to A-Cast ``{\tmsamp{Reconstruct
  Enabled}}''. Let $M$ be the set of processes, $k$, for which $k \in A_l$ for
  at least $t + 1$ processes $l \in S_i$. We now show that all processes
  in $M$ satisfy the properties of the lemma.
  
  It is clear that $M \subseteq H_i$. Thus process $i$ has received
  ``{\tmsamp{attach}} $T_j$ {\tmsamp{to}} $j$'' for every $j \in M$. Since $i$
  is assumed to be correct, the first part of the lemma is proved.
  
  We now prove the second part. First $\widetilde{S_j}$ contains $n - t
  \geqslant 2 t + 1$ processes. Now if $k' \in M$ then $k'$ belongs to $A_l$ for
  at least $t + 1$ processes $l \in S_i$. This ensures that there is
  at least one process $l$ which belongs to $\widetilde{S_j}$ as well as
  $S_i$. Now $l \in \widetilde{S_j}$ implies that $j$ has ensured that
  $A_l \subseteq \widetilde{H_j}$. Consequently, $k' \in \widetilde{H_j}$.
  
  It remains to show that $| M | \geqslant \frac{n}{3}$. We use a counting argument
  for this purpose. Let $h = | H_i |$. We have $h \geqslant n - t$.
  Consider the $h \times n$ table $\Lambda$ (relative to process $i$), where
  $\Lambda_{l, k} = \tmop{one}$ iff $i$ has received ``$l$ {\tmsamp{accepts}}
  $A_l$'' from $l$ before A-Casting ``{\tmsamp{Reconstruct Enabled}}'' and $k
  \in A_l$. Then $M$ is the set of processes $k$ such that the $k$th column in
  $\Lambda$ has at least $t + 1$ one entries. There are $n - t$ one entries in
  each row of $\Lambda$; thus there are $h (n - t)$ one entries in $\Lambda$.
  
  Let $m$ denote the minimum number of columns in $\Lambda$ that contain at
  least $t + 1$ one entries. We show that $m \geqslant \frac{n}{3}$.
  Clearly, the worst distribution of one entries in $\Lambda$ is letting
  $m$ columns be all one entries and letting each of the remaining $n
  - m$ columns have $t$ one entries. This distribution requires the
  number of one entries to be no more than $m h \noplus + (n - m) t$. Thus, we
  must have:
  \[ m h + (n - m) t \geqslant h (n - t) . \]
  This gives $m \geqslant \frac{h (n - t) - \tmop{n t}}{h - t}$. Since $h
  \geqslant n - t$ and $n \geqslant 3 t + 1$, we have
  \[ m \geqslant \frac{(n - t)^2 - n t}{n - 2 t} = n - 2 t + \frac{n t - 3
     t^2}{n - 2 t} \geqslant n - 2 t \geqslant \frac{n}{3} . \]
  This shows that $| M | \geqslant \frac{n}{3}$.
  
  \ 
\end{proof}

\begin{lemma}
  \label{icc-correctness}For every invocation of ICC[$r$], either
  \begin{itemize}
    \item For each $v \in \{ 0, 1 \}$, with probability at least $1 / 4$,
    all correct processes output $v$; or
    
    \item A set of faulty pairs will eventually be inferred by correct
    processes.
  \end{itemize}
\end{lemma}

\begin{proof}
  If $\mathbbm{E}$ occurs in any instance of IVSS[$r$] while
  executing ICC[$r$], \ then a set of faulty pairs will be inferred by
  correct processes. We prove the first part of the lemma assuming $\mathbbm{E}$
  does not occur. Suppose correct process $j$ completes ICC[$r$] with respect
  to $(\widetilde{H_k}, \widetilde{S_k})$.
  Since $\mathbbm{E}$
  does not occur, by Lemma \ref{icc-associate}, for every process $i$
  in $\mathcal{A}_j$, there is a fixed value
  $v_i$ that is distributed uniformly and independently over $[0, \ldots, u -
  1]$. Now we consider two cases:
  \begin{itemize}
    \item Let $M$ be the set of
    processes discussed in the lemma above. Clearly if $v_i = 0$ for some
    $i \in M$, then all correct processes associate $0$ with $j$ and output
    $0$. The probability that at least one process $i \in M$ has $v_i =
    0$ is $1 - \left( 1 - \frac{1}{u} \right)^{| M |}$. Since $u = \lceil 0.87
    n \rceil$, $n \geqslant 4$, and $| M | \geqslant \frac{n}{3}$ by Lemma
    \ref{icc-m}, we have $1 - \left( 1 - \frac{1}{u} \right)^{| M |} \geqslant
    1 - e^{- 0.29} \geqslant 0.25$. This implies that all correct processes output
    $0$ with probability at least $1 / 4$.
    
    \item If no process
    $i$ has $v_i = 0$ (and all correct process associate $v_i$ with $i$),
    then all correct processes output $1$. The probability of this event
    is at least $\left( 1 - \frac{1}{u} \right)^n \geqslant e^{- 1.15}
    \geqslant 0.25$.
  \end{itemize}
\end{proof}

Hence we have the following theorem.

\begin{theorem}
  Protocol ICC[$r$] is a terminating, $t$-resilient inferable common coin
  protocol.
\end{theorem}

\begin{proof}
  The termination properties follow from Lemma
  \ref{icc-termination}. The correctness properties follow from
  Lemma \ref{icc-correctness}.
\end{proof}

\section{From Common Coin to Byzantine Agreement}\label{app-aba}

First we recall a voting protocol called \tmsamp{Vote} from
{\cite{canetti1996studies}} which is a primitive required for
the construction of our ABA protocol. Protocol \tmsamp{Vote}
computes whether a detectable majority for some value among the (binary) inputs of all processes.
The output of protocol \tmsamp{Vote} is a tuple with the following meanings.

\begin{itemize}
  \item For $\sigma \in \{1, 2\}$, output $(\sigma, 2)$ means that there is an
    overwhelming majority for $\sigma$.

  \item For $\sigma \in \{1, 2\}$, output $(\sigma, 1)$ means that there is a
    distinct majority for $\sigma$.

  \item $(\bot, 0)$ means that there is no distinct majority.
\end{itemize}

\code{
  {{\tmem{{\tmstrong{Vote protocol}}}}: code for process $i$ with binary input $x_i$}
  }{
  \begin{enumeratenumeric}
    \item A-Cast ``$\text{{\tmsamp{input}}}, j, x_j$''.
    
    \item Define a set $\mathcal{A}_i$. Add $(j, x_j)$ to $\mathcal{A}_i$ if
    ``$\text{{\tmsamp{input}}}, j, x_j$'' is received from the A-Cast of process
    $j$.
    
    \item Wait until $| \mathcal{A}_i | = n - t$. Then assign $A_i
    =\mathcal{A}_i$. Set $a_i$ to the majority bit among$\{ x_j : (j, x_j) \in
    A_i \}$ and A-Cast ``$\text{\tmsamp{vote}}, i, A_i, a_i$''.
    
    \item Define a set $\mathcal{B}_i$. Add $(j, A_j, a_j)$ to $\mathcal{B}_i$
    if ``$\text{\tmsamp{vote}}, j, A_j, a_j$'' is received from the A-Cast of
    process $j$, $A_j \subset \mathcal{A}_i$, and $a_j$ is the majority bit of
    $A_j$.
    
    \item Wait until $| \mathcal{B}_i | = n - t$. Then assign $B_i
    =\mathcal{B}_i$. Set $b_i$ to the majority bit among $\{ a_j : (j, A_j,
    a_j) \in B_i \}$ and A-cast ``$\text{{\tmsamp{revote}}} \nocomma, i, B_i,
    b_i$''.
    
    \item Define a set $C_i$. Add $(j, B_j, b_j)$ to $C_i$ if
    ``$\text{{\tmsamp{revote}}}, j, B_j, b_j$'' is received from the A-cast of
    process $j$, $B_j \subset \mathcal{B}_i$, and $b_j$ is the majority bit of
    $B_j$.
    
    \item Wait until $| C_i | \geqslant n - t$. If all processes $j \in
    B_i$ had the same vote $a_j = \sigma$, then output $(\sigma, 2)$ and
    terminate. Otherwise, if all processes $j \in C_i$ have the same
    {\tmsamp{revote}} $b_j = \sigma$, then output $(\sigma, 1)$ and terminate.
    Otherwise, output $(\bot, 0)$ and complete the protocol.
  \end{enumeratenumeric}
}

This voting protocol is identical to that of {\cite{canetti1996studies}}.
The readers may refer to lemmas
5.32-5.35 {\cite{canetti1996studies}} for complete proofs.

\begin{lemma}
  All correct processes complete the voting protocol in constant time.
\end{lemma}

\begin{lemma}
  \label{vote-validaty}If all correct processes have input $\sigma$,
  then all correct processes output $(\sigma, 2)$.
\end{lemma}

\begin{lemma}
  \label{vote-level2}If some correct process outputs $(\sigma, 2)$, then
  every correct process outputs either $(\sigma, 2)$ or $(\sigma, 1)$.
\end{lemma}

\begin{lemma}
  \label{vote-level1}If some correct process outputs $(\sigma, 1)$, and no
  correct process outputs $(\sigma, 2)$, then every correct process
  outputs either $(\sigma, 1)$ or $(\bot, 0)$.
\end{lemma}

Given the voting protocol and our ICC[$r$] protocol, we can design our ABA protocol
following {\cite{canetti1993fast}}.

\code{
  {{\tmstrong{{\tmem{ABA protocol}}}}: code for process $i$ with binary input $x_i$}
}{
  \begin{enumeratenumeric}
    \item Set $r = 0$ and $v_1 = x_i$. Start the \tmem{certification} protocol.
    
    \item Repeat until completing: (each iteration is consider as a round)
    \begin{enumeratealpha}
      \item Set $r = r + 1$. Set $(y_r, m_r) = \tmop{Vote} (v_r)$.
      
      \item Invoke ICC[$r$] and wait until completion. Let $c_r$ be the
      output of ICC[$r$].
      
      \item Consider the following cases:
      \begin{enumerateromancap}
        \item If $m_r = 2$, set $v_{r + 1} = y_r$ and A-Cast
        ``{\tmsamp{complete with }}$v_r$''. Participate in only one more
        instance of the voting protocol and only one more ICC[$r$] protocol.
        
        \item If $m_r = 1$, set $v_{r + 1} = y_r$.
        
        \item Otherwise, set $v_{r + 1} = c_r$.
      \end{enumerateromancap}
      \item Upon receiving $t + 1$ ``{\tmsamp{complete with }}$\sigma$''
      A-Casts for some value $\sigma$, output $\sigma$ and complete the
      protocol.
    \end{enumeratealpha}
  \end{enumeratenumeric}
}

We now state and prove the following lemmas which are slight variants of
lemmas 5.36-5.39 presented in {\cite{canetti1996studies}}.

\begin{lemma}
  If all correct processes are in rounds greater than or equal to $r$, then every
  correct process eventually completes ICC[$r$].
\end{lemma}
\begin{proof}
  If some correct process is in a round greater than $r$, then it must have completed
  ICC[$r$]. Then by termination property (2) of ICC[$r$], every correct process
  eventually completes ICC[$r$].

  If all correct processes are in round $r$, Suppose that no correct
  process will complete ICC[$r$]. Since no correct process completes ICC[$r$],
  all correct processes keep participating. Then by termination property (1) of
  ICC[$r$], every correct process eventually completes. This is a contradiction.

  Therefore, the lemma is proved.
\end{proof}

\begin{lemma}
  In our ABA protocol, if all correct processes have the same input $\sigma$, then all
  correct processes complete and output $\sigma$.
\end{lemma}

\begin{proof}
  If all correct processes have the same input $\sigma$, then by Lemma \ref{vote-validaty} every
  correct process will output $(y_1, m_1) = (\sigma, 2)$ by the end of Step
  a. Therefore, every correct process A-Casts ``{\tmsamp{complete with
  }}$\sigma$'' in the first iteration. Therefore, every correct process will
  receive at least $n - t$ ``{\tmsamp{complete with }}$\sigma$'' A-Casts, and
  at most $t$ ``{\tmsamp{complete with }}$\sigma'$'' A-Casts. Consequently,
  every correct process will output $\sigma$.
\end{proof}

\begin{lemma}
  \label{aba-validaty}In our ABA protocol, if a correct process completes with
  output $\sigma$, then all correct processes will complete with output
  $\sigma$.
\end{lemma}

\begin{proof}
  Let us first show that if a correct process A-Casts ``{\tmsamp{complete
  with }}$\sigma$'' for some value $\sigma$, then all correct processes will
  A-Cast ``{\tmsamp{complete with }}$\sigma$''. Let $k$ be the
  first round when a correct process $i$ A-Casts ``{\tmsamp{complete with
  }}$\sigma$''. By Lemma \ref{vote-level2}, every correct process $i$ has
  $y_k = \sigma$ and either $m_k = 2$ or $m_k = 1$. Therefore, no correct
  process A-Casts ``{\tmsamp{complete with }}$\sigma'$'' at iteration $k$.
  Furthermore, all correct processes invoke the voting protocol in round $k +
  1$ with input $\sigma$. Lemma \ref{vote-validaty} now implies that, by the
  end of Step a of round $k + 1$, every correct process has $(y_{k + 1},
  m_{k + 1}) = (\sigma, 2)$. Thus, all correct processes A-Cast
  ``{\tmsamp{complete with }}$\sigma$'', either at round $k$ or at round $k +
  1$.
  
  Now assume a correct process completes with output $\sigma$. Thus, at
  least one correct process A-casted ``{\tmsamp{complete with }}$\sigma$''.
  Consequently, all correct processes A-Cast ``{\tmsamp{complete with
  }}$\sigma$''. Hence, every correct process will receive at least $n - t$
  ``{\tmsamp{complete with }}$\sigma$'' A-Casts and at most \ $t$
  ``{\tmsamp{complete with }}$\sigma'$'' A-Casts. Therefore, every correct
  process will output $\sigma$.
\end{proof}

\begin{lemma}
  If all correct processes have initiated and completed some round $k$, then
  with probability at least $1 / 4$, all correct processes have the same
  value for $v_{k + 1}$ or a set of faulty pairs will eventually be inferred
  by correct processes.
\end{lemma}

\begin{proof}
  We have two cases here. If all correct processes execute Step III in
  round $k$, then all correct processes set their $v_{k + 1}$ to the
  output of ICC[$r$]. According to the correctness property of ICC[$r$], the lemma is
  true.
  
  Otherwise, some correct process has set $v_{k + 1} = \sigma$ for some
  $\sigma \in \{ 0, 1 \}$, either in Step I or Step II of round $k$. By Lemma
  \ref{vote-level1}, no correct process will set its $v_{k + 1}$ to
  $\sigma'$. According to the correctness property of ICC[$r$], with probability at
  least $1 / 4$, all correct processes have output $\sigma$ or a
  set of faulty pairs will eventually be inferred by correct processes.
\end{proof}

\begin{lemma}
  Let $n = 3 t + \delta$, then all correct processes complete the
  ABA protocol in expected running time $O ( \frac{t}{\delta} )$.
\end{lemma}

\begin{proof}
  We first show that all correct processes complete protocol ABA within
  constant time after the first correct process initiates a
  ``{\tmsamp{complete with }}$\sigma$'' A-Cast in Step III of the protocol.
  Assume the first correct process initiates a
  ``{\tmsamp{complete with }}$\sigma$'' A-Cast in round $k$. Then all
  correct processes participate in the voting and common coin protocols of all
  the rounds up to round $k + 1$. We have seen in the proof of Lemma
  \ref{aba-validaty} that all correct processes will A-Cast
  ``{\tmsamp{complete with }}$\sigma$'' in round $k + 1$. All these A-Casts
  complete in constant time. Then every correct process completes the ABA protocol
  after completing $t + 1$ of these A-Casts. Consequently, once the first
  correct process A-Casts ``{\tmsamp{complete with }}$\sigma$'', the ABA
  protocol completes in constant time.
  
  Let the random variable $\tau$ count the number of rounds until the
  first correct process A-Casts ``{\tmsamp{complete with }}$\sigma$''. We
  have
  \[ \tmop{Prob} (\tau > k) = \tmop{Prob} (\tau \neq 1) \cdot \tmop{Prob}
     (\tau \neq 2 | \nobracket \tau \neq 1) \ldots \cdot \tmop{Prob} (\tau
     \neq k | \nobracket \tau \neq 1 \cap \ldots \cap \tau \neq k - 1) . \]
  If event $\mathbbm{E}$ does not occur in round $k$, we have
  $\tmop{Prob} (\tau \neq k | \nobracket \tau \neq 1 \cap \ldots \cap \tau
  \neq k - 1) \leqslant \frac{3}{4}$. Hence, by Lemma \ref{lemma-error-rounds},
  $\tmop{Prob} (\tau > k) \leqslant \left( \frac{3}{4} \right)^{k - 3 t /
  \delta - 1}$. By a simple calculation, we have $E (\tau) \leqslant \frac{3
  t}{\delta} + 17$. Therefore, the expected running time is $O (
  \frac{t}{\delta} )$.
\end{proof}

We have thus shown the following:

\begin{theorem}
  If $n = 3 t + \delta$, then there is an almost-surely terminating ABA
  protocol with expected running time $O ( \frac{t}{\delta} )$.
\end{theorem} 

\end{document}